\documentclass[a4paper,11pt,reqno]{amsart}
\usepackage{amsmath,amsthm,amssymb}
\usepackage{relsize}
\usepackage[noadjust]{cite}
\usepackage{color}
\usepackage{graphicx}
\usepackage[dvipsnames]{xcolor}
\usepackage{multirow}
\usepackage{booktabs}
\usepackage{verbatim} 

\usepackage{tikz}
\usetikzlibrary{decorations.pathreplacing}

\usepackage[colorlinks, urlcolor=blue, linkcolor=Purple, citecolor=red]{hyperref}
\usepackage[noabbrev,capitalize,nameinlink]{cleveref}

\usepackage{setspace}
\usepackage{enumitem}
\setitemize{itemsep=-1pt}
\setenumerate{itemsep=-1pt}

\usepackage[margin=2.4cm]{geometry}

\usepackage{stmaryrd}
\usepackage{longtable}

\usepackage{array}
\newcolumntype{x}[1]{>{\centering\arraybackslash\hspace{0pt}}p{#1}}

\theoremstyle{definition}
\newtheorem{theorem}{Theorem}[section]
\newtheorem{definition}[theorem]{{{Definition}}}
\newtheorem{ex}[theorem]{{{Example}}}

\newtheorem{remark}[theorem]{{{Remark}}}

\newtheorem{corollary}[theorem]{{{Corollary}}}
\newtheorem{proposition}[theorem]{{{Proposition}}}

\newcommand{\numberset}{\mathbb}

\newcommand{\F}{\numberset{F}}

\newcommand{\Mat}{\mbox{Mat}}

\newcommand{\C}{\mathcal{C}}

\newcommand{\U}{\mathcal{U}}

\newcommand{\wt}{\textnormal{wt}}

\newcommand{\Fq}{\F_q}

\newcommand{\mS}{\mathcal{S}}

\newcommand{\rk}{\textnormal{rk}}

\newcommand{\dr}{\dd_{\rk}}
\newcommand{\dS}{\dd_{\mathrm{S}}}

\DeclareMathOperator{\GL}{GL}

\DeclareMathOperator{\PG}{PG}

\DeclareMathOperator{\rowsp}{rowsp}

\DeclareMathOperator{\dd}{d}
\DeclareMathOperator{\Gr}{Gr}
\DeclareMathOperator{\N}{N}

\newcommand{\bF}{\mathbb{F}}
\newcommand{\bP}{\mathbb{P}}
\newcommand{\cC}{\mathcal{C}}
\newcommand{\cF}{\mathcal{F}}

\newcommand{\cP}{\mathcal{P}}
\newcommand{\cS}{\mathcal{S}}
\newcommand{\cU}{\mathcal{U}}
\newcommand{\cW}{\mathcal{W}}

\title{Schubert Subspace Codes}

\usepackage[foot]{amsaddr}
\author{Gianira N. Alfarano$^1$}
\author{Joachim Rosenthal$^2$}
\author{Beatrice Toesca$^2$}
\address{$^1$University College Dublin, Ireland.}
\address{$^2$University of Zurich, Switzerland.}
\email{gianira.alfarano@gmail.com, rosenthal@math.uzh.ch, beatrice.toesca@math.uzh.ch}

\begin{document}

\begin{abstract}
    In this paper, we initiate the study of constant dimension subspace codes restricted to Schubert varieties, which we call Schubert subspace codes. These codes have a very natural geometric description, as objects that we call intersecting sets with respect to a fixed subspace. We provide a geometric construction of maximum size constant dimension subspace codes in some Schubert varieties with the largest possible value for the minimum subspace distance. Finally, we generalize the problem to different values of the minimum distance.
\end{abstract}

\maketitle

\section{Introduction}
In algebraic coding theory, many tools from algebraic geometry have been used
in order to construct optimal codes whose codewords are all far apart under some metric. 

An important starting point for these connections was made by Goppa, who showed how algebraic curves over a finite field with many rational points and low genus lead 
to linear block codes having a comparatively large distance; see~\cite{go81}.

Recall that the {\em Grassmann variety} or \emph{Grassmannian} $\Gr(k,V)$  consists of all $k$-dimensional linear subspaces of some vector space $V$. It is well known that when the vector space $V$ has finite dimension $n$, then  $\Gr(k,V)$ is a smooth projective variety of dimension $k(n-k)$. When the base field is the finite field $\Fq$ and the vector space is $V=\Fq^n$, we will use the notation  $\Gr_q(k,n)$ to denote the Grassmannian $\Gr(k,\Fq^n)$.

It was Ryan who first showed how to construct good linear codes over the binary field $\F_2$ using the Grassmannian $\Gr_2(k,n)$; see~\cite{ry87}. Such codes are now usually referred to as {\em Grassmann codes}. Later, Nogin used the concept of a projective system to generalize the definition of Grassmann codes to arbitrary finite fields; see~\cite{no96}. He also established that the minimum distance of Grassmann codes is remarkably large, namely $q^d$, where $d$ is the dimension of the Grassmannian $\Gr_q(k,n)$. In addition, Nogin established bounds for the higher weights of Grassmann codes.

Ghorpade and Lachaud, then, used the theory of hyperplane sections of Grassmannians to come up with an elegant derivation of Nogin's results; see~\cite{gh01}.
The geometric approach also allowed the derivation of a bound for the number of MDS codes of a certain size.

Important subvarieties of the Grassmann variety are the so-called Schubert varieties. 
We will give the exact definition of a Schubert variety in the next section. 

In the fundamental paper~\cite{gh05}, Ghorpade and Tsfasman used the language of projective systems to describe linear codes in a natural way starting with some projective variety. 
This then allowed them to generalize the notion of Grassmann codes to the notion of a
{\em Schubert code}, opening up many interesting research questions relating Schubert codes to combinatorial questions arising also in Schubert calculus. The interested reader is referred to the paper by Ghorpade and Singh~\cite{gh18}, where many properties of Schubert codes are described by geometric means.

There is a second natural connection between Grassmann varieties and coding theory. 
As shown by K\"otter and Kschischang in \cite{ko08}, it is natural in the area of network coding to consider the elements of the Grassmannian $\Gr_q(k,n)$  as codewords of a so-called constant dimension subspace code. In this case, the distance between codewords is not the Hamming distance but rather the subspace distance, which we will define in the next section. One immediately has a lot of interesting coding questions that are geometrically related to the geometry of the Grassmann variety. 

As it is natural (and interesting!) to go from Grassmann codes to Schubert codes, we demonstrate in this paper that it is interesting to consider subspace codes which are all contained in a fixed Schubert variety.
\bigskip

\paragraph{\textbf{Our Contribution}} In this paper we introduce and study \emph{Schubert subspace codes}, i.e. constant dimension subspace codes in the Grassmannian $\Gr_q(k,rk)$, restricted to a given Schubert variety. We introduce the geometric concept of \emph{intersecting set with respect to a given subspace} and we show that studying these objects is equivalent to constructing Schubert subspace codes with bounded minimum distance. We focus mostly on codes in $\Gr_q(k,rk)$ with the largest possible minimum subspace distance, i.e.\ $2k$. Using the theory of linear sets, we provide a geometric construction of such codes with maximum size when $u\leq \frac{rk}{2}$. 
In this case, we compare our construction with the multilevel one, proposed by Etzion and Silberstein in \cite{etzion2009error}.  
We show that the latter does not give a maximum size subspace code in case $r<u$. Hence, in this case, our construction provides a larger code. On the contrary, when $r\geq u$, both constructions can achieve the maximum possible size. We generalize the problem in order to allow different values for the minimum distance. However, in this case, things get more involved and complicated to handle.
\bigskip

\paragraph{\textbf{Outline}} The paper is organized as follows. In \cref{sec:prel} we provide the necessary background for the rest of the paper. In \cref{sec:prob_formulation} we introduce the problem of finding bounds on the size $m_q$ of constant dimension subspace codes restricted to Schubert varieties and establish the first results and constructions. Moreover, we compare our constructions with the multilevel one. In \cref{sec:generalization} we explain a more general problem and we provide a lower bound on $m_q$ given by the multilevel construction. We conclude in \cref{sec:conclusion} with several open questions.

\section{Preliminaries}\label{sec:prel}

In this section, we provide the background material for the rest of the paper. In particular, we present some well-known facts about rank-metric codes, Ferrers diagrams, Schubert varieties and subspace codes.

Throughout the paper, $\F_q$ denotes the finite field of order $q$, and $k,n\geq 2$ are fixed positive integers, with $k\leq n$. 
We denote by $\Mat(\ell,m,\F_q)$ the space of $\ell \times m$ matrices with entries in $\F_q$ and, in case of square $\ell\times \ell$ matrices, we simply write $\Mat(\ell,\F_q)$.
For a matrix $M\in \Mat(\ell,m,\F_q)$ we denote by $\mathrm{rowsp}(M)$ the rowspace of $M$ over $\F_q$, that is the $\F_q$-subspace of $\F_q^m$ generated by the rows of $M$. 
We denote by $\bP_q^{t-1}$ the $(t-1)$-dimensional projective space with underlying vector space $\bF_q^{t}$. Notice that this corresponds to what is often denoted as $\PG(t-1,q)$ in the finite geometry literature. We further denote by $\Gr_q(k,n)$ the Grassmannian, i.e.\ the set of all the $k$-dimensional subspaces of~$\F_q^{n}$, and by $\binom{n}{k}_q$ the Gaussian binomial coefficient representing its cardinality. Note that $\bP_{q}^{n-1}=\Gr_q(1,n)$. Finally, given $V\in\Gr_q(k,n)$, we denote by $\bP(V)$ its associated projective subspace.

\subsection{Rank-Metric Codes}\label{subsec:rank-metric}
Rank-metric codes were originally introduced by Delsarte in~\cite{de78} and have been intensively investigated in recent years because of their applications in network coding \cite{silva2008rank}. 
We endow the space $\Mat(n, m,\F_q)$ with the \textbf{rank metric}, defined as
\[\dr(A,B) = \mathrm{rk}\,(A-B), \qquad \mbox{ for every } A,B\in\Mat(n, m,\F_q).\]
 An $[n\times m,M]_q$ \textbf{rank-metric code} $\C$ is a subset of $\Mat(n,m,\F_q)$ of size $M$. The \textbf{minimum rank distance} of $\C$ is defined as
 \begin{align*}
     \dr(\C):= &\min\{ \dr(A,B) \mid A,B \in \C,\,\, A\neq B \}.
 \end{align*}
If $\dr(\C)=\delta$ is known, we write that $\C$ is an $[n\times m,M,\delta]_q$ code.
The parameters $n,m,M,\delta$ of an $[n\times m,M,\delta]_q$ rank-metric code satisfy a Singleton-like bound \cite{de78}, that reads~as
\[ M \leq q^{\max\{m,n\}(\min\{m,n\}-\delta+1)}. \]
When equality holds, we say that $\C$ is a \textbf{maximum rank distance} code, or \textbf{MRD} for short.

If a rank-metric code $\C$ is a linear $\F_q$-subspace of $\Mat(n,m,\F_q)$, then its cardinality is $M=q^k$, where $k$ is the $\F_q$-dimension of $\C$. In this case, we write that $\C$ is an $[n\times m,k]_q$ \textbf{linear rank-metric code}.
Furthermore, if  $\C$ is in addition $\F_{q^m}$-linear, we can represent it as an $\F_{q^m}$-subspace of $\F_{q^m}^n$. If $\dim_{\F_{q^m}}(\C)=k$, then we write that $\C$ is an $[n,k]_{q^m/q}$ \textbf{linear} rank-metric code. In this case, $\C=\rowsp(G)$ for some $G\in\Mat(k,n,\F_{q^m})$ called \textbf{generator matrix} of $\C$. If the columns of one (and hence any) generator matrix of $\C$ are $\F_q$-linearly independent, then $\C$ is said to be \textbf{nondegenerate}. 
Two $[n,k]_{q^m/q}$ rank-metric codes $\C$ and $\C'$ are \textbf{equivalent} if and only if
there exists $A \in \mathrm{GL}(n,\F_q)$ such that
$\C'=\C A=\{vA : v \in \C\}$.

There is a geometric interpretation of $[n,k]_{q^m/q}$ linear rank-metric codes as $q$-\textit{systems}. This concept will be useful for the construction in \cref{sec:prob_formulation}.

\begin{definition}
An $[n,k]_{q^m/q}$ \textbf{system} $\U$ is an $\F_q$-subspace of $\F_{q^m}^k$ of dimension $n$.
Moreover, two $[n,k,]_{q^m/q}$ systems $\U$ and $\U'$ are \textbf{equivalent} if there exists an $\F_{q^m}$-isomorphism $\phi\in\GL(k,\F_{q^m})$ such that $ \phi(\U) = \U'$.
If the parameters are clear from the context or not relevant, we simply say that $\U$ is a $q$-\textbf{system}.
\end{definition}

In \cite{sheekey2019scatterd, randrianarisoa2020geometric} it has been proven that there is a one-to-one correspondence between equivalence classes of nondegenerate $[n,k]_{q^m/q}$ codes and equivalence classes of $[n,k]_{q^m/q}$ systems and the correspondence can be explained as follows.
Let $\C$ be an $[n,k]_{q^m/q}$ code and $G$ be a generator matrix for $\C$. Define $\U$ to be the $\F_q$-span of the columns of $G$. In this case, $\U$ is also called a \textbf{system associated with} $\C$. Vice versa, given an $[n,k]_{q^m/q}$ system $\U$, define $G$ to be the matrix whose columns are an $\F_q$-basis of $\U$ and let $\C$ be the code generated by $G$. $\C$ is also called a \textbf{code associated with} $\U$. For more details on this correspondence, see also \cite{alfarano2022linear}.

\subsection{Ferrers Diagram Codes}\label{subsec:ferrers_diagrams}
Etzion and Silberstein introduced Ferrers diagram rank-metric codes in \cite{etzion2009error} as a tool to construct subspace codes of large size. An $n\times m$ \textbf{Ferrers diagram} $\cF$ arises in combinatorics as a way to represent partitions of an integer through patterns of dots. In particular, $n$ right-justified rows of dots are displayed in non-increasing order (from top to bottom) and the first row has $m$ dots.
Given an $n\times m$ Ferrers diagram $\cF$, we say that a code $\cC$ is an $[\cF, M, \delta]_q$ \textbf{Ferrers diagram rank-metric code} if it is an $[n\times m, M, \delta]_q$ rank-metric code and all entries not indexed by $\cF$ are zeroes. 

For a given Ferrers diagram $\cF$, and for every $0 \leq i \leq \delta -1$ we denote by $\nu_i$ the number of dots in $\cF$ that are not contained in the first $i$ rows nor in the rightmost $\delta-1-i$ columns. Denote by $\nu_{\min}(\cF,\delta)=\min_{0\leq i\leq \delta-1}\{\nu_i\}$. Then the Singleton-like bound for 
the parameters of a $[\cF, M, \delta]_q$ Ferrers diagram rank-metric code reads as
\begin{equation} \label{eq:Ferrers_bound}
    M\leq q^{\nu_{\min}(\cF,\delta)}.
\end{equation}
When the bound is attained with equality, we say that $\cC$ is a \textbf{maximum Ferrers diagram} code (or \textbf{MFD} code for short). In \cite{etzion2009error}, the authors conjectured that it is possible to construct an MFD code for every given set of parameters $q,\delta,$ and $\cF$. This statement, known as the Etzion-Silberstein conjecture, has been proven for many classes of Ferrers diagrams and sets of parameters, but the conjecture is still widely open; see \cite{neri2023proof} for some recent advances towards the proof of the conjecture and references therein for the known constructions.

\subsection{Schubert Varieties} \label{subsec:schubert_varieties}
Schubert varieties and Schubert cells can be defined as subsets of a Grassmannian $\Gr_q(k,n)$, where we only consider the vector subspaces that intersect the elements of a given flag in subspaces of fixed dimension. More precisely, let $F$ be a \textbf{flag}, i.e.\ a sequence of nested linear subspaces $\{0\}\subset V_1 \subset V_2 \subset \dots \subset V_n=\bF_q^n$ such that $\dim(V_j)=j$, for every $j=1,2,\dots,n$. Let $d=(d_1,\dots,d_k)$ be an ordered index set satisfying $1\leq d_1 < \dots < d_k \leq n$. We say that a space $W\in\Gr_q(k,n)$ satisfies the \textbf{Schubert condition} $d$ if for every $i=1,\dots,k$, $V_{d_i}$ intersects $W$ in dimension at least $i$.

\begin{definition}
    Let $F=(V_1,V_2,\dots,V_n)$ be a flag and $d=(d_1,\dots,d_k)$ be a Schubert condition.
    The \textbf{Schubert variety} $\Omega_d$ is the set of all subspaces satisfying the Schubert condition $d$, i.e., 
        $$ \Omega_d := \{W\in\Gr_q(k,n) \mid \dim(W\cap V_{d_i})\geq i, \; \; \forall i=1,\dots,k \}.$$
    The \textbf{Schubert cell} $\Omega^0_d$ is the set of all subspaces such that for every $i$, $V_{d_i}$ is the first element of the flag that intersects $W$ in a space of dimension $i$, i.e.,
    \begin{align*}
        \Omega_d^0 :&= \{W\in\Gr_q(k,n) \mid \dim(W\cap V_{d_i})= i,\textnormal{ and } \dim(W\cap V_{d_i-1})=i-1, \; \; \forall i=1,\dots,k \}\\
        &= \{W\in \Omega_d \mid \dim(W\cap V_{d_i-1})=i-1, \; \; \forall i=1,\dots,k \}.
    \end{align*}   
\end{definition}

Notice that $\Gr_q(k,n)\cong \{A\in\Mat(k,n,\F_q) \mid \rk(A)=k\}/\GL(k,\F_q)$. In other words,  $\Gr_q(k,n)$ can be represented by the set of $k\times n$ matrices of rank $k$ in \textbf{reduced row echelon form} (RREF). Using this representation, we can naturally partition $\Gr_q(k,n)$  according to the pivot positions of the matrices in RREF.
Let us now consider the \textbf{standard flag}, i.e.\ the one where the subspaces are defined as $V_i=\mathrm{span}\{e_{n-i+1}, \dots, e_n\}$, with $e_i$ the vector of the standard basis having a one in position $i$ and zero in all other entries.
Then, each of the sets in the partition of the Grassmannian corresponds to a different Schubert condition with respect to the standard flag. In particular, the vector $p=(p_1,\dots,p_k)$ whose entries are the positions of the $k$ pivots in the RREF is related to the Schubert condition $d=(d_1,\dots,d_k)$ with respect to the standard flag by the equations
\begin{equation} \label{eq:rel_pi_di}
	p_i = n+1-d_{k+1-i} \quad \forall i=1,\dots,k.
\end{equation}

Notice that by letting $\GL(n,\bF_q)$ act on $\Mat(k,n,\bF_q)$ from the right, we can change the basis over $\bF_q^n$ and can therefore send any given flag to the standard one. Hence, we will only consider Schubert varieties with respect to the standard flag, as in this way we have a convenient representation of subspaces as matrices in reduced row echelon form and can use explicit coordinates.
For example, with these standard coordinates, the largest Schubert cell in the partition, also known as \textbf{thick open cell} and denoted by $\Gr_q^0(k,n)$, is the one where the pivots of the RREF are all in the first $k$ entries, and is associated to the Schubert condition $d=(n-k+1,\dots,n)$. 

\subsection{Subspace Codes}\label{subsec:subspace_codes}
In this subsection, we recall the definition of subspace codes and explain a connection between the Schubert cells of the Grassmannian and Ferrers diagram rank-metric codes. For a detailed treatment of the topic, we refer the interested reader 
to the surveys ~\cite{ho18,kurz2021construction}.

The Grassmannian $\Gr_q(k,n)$ can be interpreted as a metric space endowed with the \textbf{subspace distance}, defined for every $U,V\in\Gr_q(k,n)$ as
\begin{align*} \dS(U,V) &= \dim_{\F_q}(U+V) -\dim_{\F_q}(U\cap V)= 2(k-\dim_{\bF_q}(U\cap V)).
\end{align*}

A \textbf{constant dimension subspace code} is a family of subspaces $\cS\subseteq \Gr_q(k,n)$. The minimum subspace distance $\dd_\mathrm{S}(\cS)$ of $\mS$ is defined as 
$$\dd_{\mathrm{S}}(\cS)=\min\{\dd_{\mathrm{S}}(U,V) \mid U,V\in \cS, \: U\ne V\}.$$

It is clear that the largest possible minimum distance value that a subspace code $\cS$ can achieve is $2k$, which arises when $n\geq 2k$ and all the subspaces in $\cS$ intersect trivially. Subspace codes that have distance $2k$ and are maximal in the sense that every vector $v\in \Fq^n $ is inside some codeword $U\in\cS$, define geometrically a spread. For this reason, such codes have been called \textbf{spread codes}
in the coding literature and questions like efficient decoding of this class of codes have been studied; see~\cite{go12,ma08p}.

Partitioning the Grassmannian into Schubert cells based on the pivot positions in the RREF, one finds that the thick open cell can be described as follows:
$$\Gr_q^0(k,n)=\{U\in\Gr_q(k,n) \mid  U=U_A=\rowsp(\mathrm{Id}_{k} \mid A ), \textnormal{ for }A\in\Mat(k,n-k,\F_q)\}.$$
In this cell, we have that $\dd_{\mathrm{S}}(U_A,U_B) = 2\cdot\rk(A-B)$. In other words the metric spaces $(\Gr_q^0(k,n),\dd_{\mathrm{S}})$  and $(\Mat(k,n-k,\F_q),2\cdot\dd_{\rk})$ are isometric.

We can define a similar isometry for subspaces that belong to different Schubert cells.
Let $d$ be a Schubert condition with respect to the standard flag and $U\in\Omega_d^0\subseteq\Gr_q(k,n)$. Consider the representation of $U$ in RREF as a $k\times n$ matrix of rank $k$, with pivot position vector $p=(p_1,\ldots,p_k)$. By substituting all entries that are not to the left or above a leading coefficient with a dot, we obtain a matrix in the so-called \textbf{echelon-Ferrers form}. If we now consider only the pattern created by the dots, we obtain a Ferrers diagram $\cF$. Notice that there is a one-to-one correspondence between the Schubert cell $\Omega_d^0$ and the Ferrers diagram $\cF$. Given two matrices $A, B$ supported on $\cF$, they can be associated with the subspaces $U_A, U_B\in\Omega_d^0$, where the pivot positions are determined by the Ferrers diagram's shape and the other entries by the matrices $A,B$. We call this procedure \textbf{lifting}. It holds that $\dS(U_A,U_B)=2\cdot\dr(A,B)$. Hence, the space $(\Omega_d^0,\dS)$ of constant dimension codes in a Schubert cell is isometric to the space $(\Mat(\cF,\bF_q),2\cdot\dr)$ of Ferrers diagram rank-metric codes supported on the corresponding Ferrers diagram $\cF$.

In \cite{etzion2009error}, in order to design constant dimension codes of fixed distance and large size, the authors introduced the \textbf{multilevel construction}, where one takes the union of constant dimension codes constructed in different Schubert cells, each of them obtained by lifting a Ferrers diagram rank-metric code. Because of the properties of the Schubert cells and the isometry within each cell, the subspace code has double the distance of the Ferrers diagram rank-metric codes that are lifted to generate it.

\section{Subspace Codes Restricted to Schubert Varieties}\label{sec:prob_formulation}
We start this section by introducing the concept of subspace codes restricted to Schubert varieties.

\begin{definition} \label{def:Schubert_code}
    Let $\Omega_d$ be a Schubert subvariety of the Grassmannian $\Gr_q(k,n)$. A family $\cS\subseteq\Omega_d$ equipped with the subspace distance is called a \textbf{Schubert subspace code} with respect to $\Omega_d$.
\end{definition}

The natural coding theory question that arises is whether we can construct Schubert subspace codes of large size, where a certain minimum subspace distance is guaranteed. In the sequel, we examine some special cases for which we are able to provide codes with good parameters.

Throughout this paper, we assume $n=rk$, with $r,k$ positive integers, and work in the Grassmannian $\Gr_q(k,rk)$.
We first define the coordinate map from $\F_{q^k}^r$ to $\F_q^{rk}$, which is sometimes also called (vector) \textbf{field reduction} map; see \cite{lavrauw2015field}. This is simply the map $\varphi$ that associates to a vector $v\in\F_{q^k}^r$ the corresponding vector $\varphi(v)\in\F_q^{rk}$ by expanding the coordinates. More precisely, let $\{w_1,\ldots,w_k\}$ be a fixed basis of $\F_{q^k}$ over $\F_q$. Then
\begin{equation*} \begin{matrix}
        \varphi :&\F_{q^k}^r &\longrightarrow &\F_{q}^{rk}\\
        &(v_1, \ldots, v_r) &\longmapsto  &(c_{11},\ldots,c_{1k},c_{21},\ldots,c_{2k},c_{r1},\ldots,c_{rk}),	\end{matrix} \end{equation*}
where $v_i=\sum_{j=1}^{k}c_{ij}w_j$ for any  $i=1,\dots,r$.

In the same way, we can also define a projective version $\overline{\varphi}$ of the field reduction map, which identifies points of $\bP_{q^k}^{r-1}$ with $(k-1)$-dimensional subspaces of $\bP^{rk-1}_{q}$.
It is well-known that $\mathrm{Im}(\varphi)$ is a (vectorial) $k$-spread of $\F_q^{rk}$, which naturally gives rise to a projective $(k - 1)$-spread of $\bP^{rk -1}_{q}$. Such a spread is known as \textbf{Desarguesian spread}; see \cite{segre1964teoria}.

It is in general a challenging coding question to construct Schubert subspace codes having large cardinality for any subspace distance and any Schubert variety $\Omega_d$.
In this section, we provide some first results in the case where the minimum distance is $2k$ and for a particular class of Schubert varieties. Let $U$ be a subspace of $\bF_q^{rk}$ of dimension $u$ and denote by $\Omega_U$ the Schubert variety
$$\Omega_U := \{W\in\Gr_q(k,rk) | \dim(W\cap U)\geq 1\}.$$
 
\begin{definition} \label{def:intersecting_set}
    Let $k,r\geq 2$, $u\leq\frac{rk}{2}$ and let $U$ be a $u$-dimensional subspace of $\bF_q^{rk}$.  We say that a Schubert subspace code $\cS\subseteq \Omega_U$ is an \textbf{intersecting set with respect to $U$} if for every $S_i,S_j\in\mathcal{S}$ such that $ S_i\neq S_j$, it holds that $S_i\cap S_j =\{0\}$.
\end{definition}

We denote by $m_q(k,r,u)$ the largest size of an intersecting set with respect to $U$.

\begin{remark} \label{rk:m_well_defined}
Note that clearly $m_q(k,r,u)$ is well-defined, i.e.\ it does not depend on the choice of $U$, but only on its dimension $u$. Moreover, observe that $\cS$ gives rise to a partial spread~\cite{be75} in the Schubert variety with Schubert condition $d=(u,rk-k+2,\dots,rk-1,rk)$.
\end{remark}

Since the $k$-dimensional subspaces of $\F_q^{rk}$ are in correspondence with the $(k-1)$-dimensional projective subspaces of $\bP^{rk-1}_{q}$, we then have an equivalent ``projective" formulation for $m_q(k,r,u)$. Indeed, we can also consider  $U$ to be a $(u-1)$-dimensional subspace of $\bP^{rk-1}_{q}$. 
In this setting, we want to find a collection $\mathcal{S}$ of disjoint $(k-1)$-dimensional subspaces of $\bP^{rk-1}_{q}$ that intersect $U$ in at least one point. The two formulations are equivalent, hence with a little abuse of notation, we might switch from one to the other if the context is clear and talk about intersecting sets also when the elements of $\cS$ are projective subspaces.

We can easily compute an upper bound for $m_q(k,r,u)$. We include the proof for the sake of completeness.

\begin{proposition}\label{thm:upp_bound}
Let $k,r\geq 2$ and let $u<rk$. Then, $$m_q(k,r,u)\leq \binom{u}{1}_q.$$ 
\end{proposition}
\begin{proof}
    The bound follows from the fact that a subspace $U\leq \F_q^{rk}$ contains $\binom{u}{1}_q=\frac{q^u-1}{q-1}$ one-dimensional subspaces. If $m_q(k,r,u)>\binom{u}{1}_q$, then at least two subspaces in $\mS$ intersect $U$ in the same one-dimensional subspace. Hence, we get a contradiction.
\end{proof}

In the following, we are going to show that when $u\leq \frac{rk}{2}$, there exist some cases for which the upper bound from \cref{thm:upp_bound} is tight. The construction we provide makes use of geometric objects known as \textit{linear sets}.
Linear sets in finite geometry can be viewed as a generalization of subgeometries. Their name was first proposed by Lunardon in \cite{lunardon1999normal}, where linear sets are used for special constructions of blocking sets. The interested reader is referred to \cite{polverino2010linear} for an in-depth treatment. 

Special linear sets, which are of particular interest for this paper, are the \emph{scattered} ones, introduced by Blokhuis and Lavrauw in \cite{blokhuis2000scattered}. More recently, Sheekey and Van de Voorde observed a connection between scattered linear sets and MRD codes in \cite{sheekey2016new,Sheekey2020}; see also \cite{polverino2020connections}.

\begin{definition}
Let $\U$ be a $[u,r]_{q^k/q}$-system. The $\F_q$-\textbf{linear set} in $\bP^{r-1}_{q^k}$ of rank $u$ associated with $\U$ is the set 
$$L_\U:=\{\langle v \rangle_{\F_{q^k}} \mid v\in \U\setminus\{0\}\} \subseteq \bP^{r-1}_{q^k},$$
where $\langle v \rangle_{\F_{q^k}}$ denotes the projective point corresponding to the vector $v$.
\end{definition}

Let $\Lambda=\bP(\cW)$ be the projective subspace corresponding to the $\F_{q^k}$-subspace~$\cW$ of~$\F_{q^k}^r$. We define the \textbf{weight} of $\Lambda$ in $L_\U$ as the integer
$$ \wt_{\U}(\Lambda):=\dim_{\Fq}(\cU \cap \cW).$$
For $P\in \bP^{r-1}_{q^k}$, we have that $P \in L_\U$ if and only if $\wt_{\U}(P)\geq 1$. 

Let $\U$ be a $[u,r]_{q^k/q}$-system. Then, the cardinality of the associated linear set $L_\U$ satisfies 
\begin{equation}\label{eq:linear_set} |L_\U|\leq \frac{q^u-1}{q-1}. \end{equation}
A linear set $L_\U$ whose cardinality meets \eqref{eq:linear_set} with equality is said to be \textbf{scattered}. Equivalently, a linear set $L_{\U}$ is scattered if and only if $\wt_{\U}(P)=1$ for each $P\in L_{\U}$. 

In \cite{blokhuis2000scattered}, Blokhuis and Lavrauw provided an upper bound on the parameters of scattered linear sets. If $\U$ is a $[u,r]_{q^k/q}$-system such that the linear set $L_\U$ is scattered, then 
\begin{equation}\label{eq:upp_bound_scattered}
    u\leq\frac{rk}{2}.
\end{equation}

During the past years, a lot of progress has been made concerning the study of \textbf{maximum scattered linear sets}, which are linear sets whose parameters meet the bound \eqref{eq:upp_bound_scattered} with equality. When $r$ is even, a construction of maximum scattered linear sets has been provided by Blokhuis and Lavrauw; see \cite{blokhuis2000scattered}. When $r$ is odd and $k$ is even a construction of maximum scattered linear sets has been provided for infinitely many parameters; see \cite[Theorem 1.2]{bartoli2018maximum}. Finally, in \cite{csajbok2017maximum}, it is shown that, whenever $rk$ is even, there exists an $\left[\frac{rk}{2},r\right]_{q^k/q}$-system $\U$, such that $L_\U$ is scattered. When $rk$ is odd, one of the few existence results on the
maximum rank of a scattered linear set is due to Blokhuis and Lavrauw; see \cite[Theorem 4.4]{blokhuis2000scattered}.
They show that there  exists an $[ab,r]_{q^k/q}$ system such that $L_{\U}$ is scattered, whenever $a$ divides $r$, $\gcd(a,k)=1$ and
$$ ab<\begin{cases}\frac{rk-k+3}{2} & \mbox{ if } q=2 \mbox{ and } a=1,\\
\frac{rk-k+a+3}{2} & \mbox{ otherwise.}
\end{cases}$$

Let $\U$ be an $[u,r]_{q^k/q}$-system such that $L_\U$ is scattered. In \cite[Lemma 6.6]{alfarano2022linear}, it is stated that if $u>r$, then there exists an $[u-1,r]_{q^k/q}$-system $\mathcal{V}$, such that $L_{\mathcal{V}}$ is scattered. 

Using these results, we are now ready to present our construction of intersecting sets with respect to a subspace $U$ of $\bF_q^{rk}$.  

\begin{proposition}\label{prop:constr_field_red}
    Let $\U$ be a $[u,r]_{q^k/q}$-system and let $L_\U$ be the linear set associated with it. Let~$\mathcal{P}$ be the set of points of weight $1$ in $L_\U$. Then $$\{\overline{\varphi}(P) \mid P\in\mathcal{P}\}\subseteq \bP^{rk-1}_{q}$$ is an intersecting set with respect to $\bP(\varphi(\U))$ (in the projective sense).
\end{proposition}
\begin{proof}
    First of all, we notice that if $P=\bP(\langle p\rangle_{\bF_{q^k}})\in\mathcal{P}$, then $\overline{\varphi}(P)$ is a $(k-1)$-dimensional subspace of $\bP^{rk-1}_{q}$. Moreover, $1=\dim_{\F_q}(\U\cap \langle p\rangle_{\bF_{q^k}})$, which implies that the intersection $\bP(\varphi(\U))\cap \overline\varphi(P)$ is non empty. Finally, for distinct $P\in\mathcal{P}$ the spaces $\overline\varphi(P)$ are clearly disjoint. Then $\{\overline{\varphi}(P) \mid P\in\mathcal{P}\}\subseteq \bP^{rk-1}_{q}$ is a collection of disjoint subspaces intersecting $\bP(\varphi(\cU))$.
\end{proof}

Clearly, the collection of spaces $\{\overline\varphi(P)\mid P\in\mathcal{P}\}$, where $\mathcal{P}$ is the set of points of weight $1$ in a linear set, is a subset of the Desarguesian spread. The following corollary is then immediate.

\begin{corollary}
Let $\U$ be an $\left[\frac{rk}{2},r\right]_{q^k/q}$-system such that the linear set $L_\U$ is maximum scattered.
Then $m_q\left(k,r,\frac{rk}{2}\right)=\frac{q^{\frac{rk}{2}}-1}{q-1}$.
\end{corollary}

Hence, given an $q$-system $\U$ such that $L_\U$ is a maximum scattered linear set, \cref{prop:constr_field_red} provides an intersecting set of maximum size with respect to $\varphi(\cU)$. 

\subsection{Special Case}

In this subsection, we study the special case where $U\subseteq\bF_q^{rk}$ is a subspace of dimension $u=k$ in \cref{def:intersecting_set}. In this case, we provide a more algebraic construction of a maximum size Schubert subspace code $\cS$ with respect to $\Omega_U$ with minimum subspace distance $2k$. We then compare it to the previous construction for general dimension $u$.

Consider the $[k,r]_{q^k/q}$-system
\begin{equation} \U:=\{(s,s^q, s^{q^2},\dots,s^{q^{r-1}})\mid s\in\bF_{q^k}\}. \label{eq:U}\end{equation}

\begin{theorem}
    Consider the field extension $\bF_{q^k}/\bF_q$ and let $\U$ be as defined in \eqref{eq:U}.
    Let $a\in\F_{q^k}$ be such that its field norm is $\N_{\F_{q^k}/\F_q}(a)=1$ and let
    \begin{equation}
    \sigma_a:=\{(s,as, s^{q^2},\ldots, s^{q^{r-2}}, s^{q^{r-1}})\mid s\in\F_{q^k}\}. \label{eq:l_a}\end{equation}
 Then $\U\cap\sigma_a$ is an $\bF_q$-linear space in $\bF_{q^k}^{r}$ of dimension 1.
\end{theorem}
\begin{proof} Observe that
    \begin{equation*}
		\U\cap\sigma_a =
		 \{(s,s^q, s^{q^2},\dots,s^{q^{r-1}})\mid s\in\bF_{q^k}\text{ s.t. }s^q=as\}.		\end{equation*}
	To determine this intersection we need to study the roots of the polynomial
    $$L_a(x)=x^q-ax=x(x^{q-1}-a).$$
	Notice that $\forall a\in\bF_{q^k}$, $L_a(x)$ is a nonzero linearized polynomial of degree $q$ over $\bF_{q^k}$. Hence, its roots form an $\bF_q$-linear subspace of $\bF_{q^k}$ of either dimension $0$, in which case the only root is $x=0$ counted with multiplicity $q$, or of dimension $1$, in which case $L_a(x)$ has $q$ distinct roots of multiplicity $1$. By Hilbert's Theorem 90 (see \cite[Ex. 2.33]{lidl1994introduction}), we have that
	\begin{equation*} \label{eq:hilbert90}
		\N_{\bF_{q^k}/\bF_{q}}(a)=1 \iff \exists \beta\in\bF_{q^k}^*\text{ s.t. }a=\frac{\beta}{\beta^q}=\beta^{1-q}.
	\end{equation*}
	If there exists such $\beta$, its inverse is also a non-zero element of $\bF_{q^k}$ and it is a root of the linearized polynomial $L_a(x)$. 
	Hence, $\N_{\bF_{q^k}/\bF_{q}}(a)=1$ implies that $L_a(x)$ has a nonzero root and, since it has degree $q$ and its roots form an $\bF_q$-linear subspace, there are $q$ distinct roots, each with multiplicity~$1$.
\end{proof}

Now, observe that in $\F_{q^k}$ there are $|\ker(\N_{\bF_{q^k}/\bF_{q}})| = \dfrac{q^k-1}{q-1}$ elements with norm equal to $1$, so there are $\dfrac{q^k-1}{q-1}$ many spaces $\sigma_a$ defined as in Equation \eqref{eq:l_a}. Hence we have the following.

\begin{corollary}\label{cor:construction_norm}
    If $\varphi$ is the field reduction map, the collection
    \begin{equation}\{\varphi(\sigma_a) \mid a\in\F_{q^k}, \; \N_{\F_{q^k}/\Fq}(a)=1\}\label{eq:sigma_a}
        \end{equation}
    forms an intersecting set with respect to the subspace $\varphi(\U)$. 
\end{corollary}

Up to isomorphism we can send $\varphi(\cU)$ to any subspace $U$ of $\F_q^{rk}$ and the collection \eqref{eq:sigma_a} to an intersecting set of maximum size $m_q(k,r,k)=\binom{k}{1}_q$ with respect to such $U$.

\begin{remark}
The $q$-system $\U$ defined in Equation \eqref{eq:U} is the $q$-system associated with a $[k,r]_{q^k/q}$ Gabidulin code. Such codes are MRD and constitute a well-studied family of rank-metric codes. Moreover, in \cite{Sheekey2020} it is shown that $L_\U$ is scattered. Hence, by considering the field reduction of all the points of $L_\U$ as in \cref{prop:constr_field_red}, we obtain another collection of disjoint subspaces intersecting  $\varphi(\U)$. The natural question that arises is whether the two constructions we provided are equivalent or, in other words, whether the spaces $\varphi(\sigma_a)$ in Equation \eqref{eq:l_a} can be regarded as field reduction of points of $L_\U$.
\end{remark}

First of all, we give the following notion of equivalence.

\begin{definition}\label{def:equivalence}
   Let $U$ and $U'$ be two $k$-dimensional subspaces $\F_{q}^{rk}$. Let $\cS_U$ and $\cS_{U'}$ be two intersecting sets with respect to $U$ and $U'$, respectively. Then $\mS_{U}, \mS_{U'}$ are \textbf{equivalent} if there exists an $\F_q$-linear map $\psi\in\GL(rk,\Fq)$ such that $\psi(U)=U'$ and $\psi(\mS_U)=\mS_{U'}$.
\end{definition}

Let $\U$ be the $[k,r]_{q^k/q}$-system $\{(s,s^q,\ldots,s^{q^{r-1}}) \mid s\in\F_{q^k}\}$, given in Equation~\eqref{eq:U} and $\cP$ be the set of points of weight 1 in $L_{\cU}$. Let $\mS_1:=\{\varphi(\sigma_a) \mid a\in\F_{q^k}, \; \N_{\F_{q^k}/\F_q}(a)=1\}$ and $\mS_2:=\{\overline\varphi(P) \mid P \in \cP\}$.  
When $r=2$, $\mS_1$ and $\mS_2$ are essentially the same construction. 
For a general $r$ though, we do not know whether $\mS_1$ and $\mS_2$ are equivalent or not in the sense of \cref{def:equivalence}.
However, we can find a map that sends $\mS_1$ in the set given by the field reduction of the points of a different linear set, as the following proposition illustrates.

\begin{theorem}\label{prop:transformation}
Let $\U$ be the $[k,r]_{q^k/q}$-system $\{(s,s^q,\ldots,s^{q^{r-1}}) \mid s\in\F_{q^k}\}$, given in Equation~\eqref{eq:U}. Let $\mS_1:=\{\varphi(\sigma_a) \mid a\in\F_{q^k}, \; \N_{\F_{q^k}/\F_q}(a)=1\}$. 
Let $\U'=\{(x,x^q,x,\ldots,x)\mid x\in\F_{q^k}\}$ and $\cP'$ be the set of points of weight 1 in $L_{\cU'}$. Then $\mS_1$ is equivalent to $\mS_2:=\{\overline\varphi(P) \mid P\in \cP'\}$.
\end{theorem}

\begin{proof}
    Define the $\F_q$-linear map
    $$\begin{array}{rccl}
    \psi:&\F_{q^k}^r &\to &\F_{q^k}^r\\
& (x_1,x_2,\ldots,x_r)&\mapsto & \left(x_1,x_2,x_3^{q^{k-2}},x_4^{q^{k-3}},\ldots,x_r^{q^{k-r+1}}\right).\end{array}$$
Clearly, $\U'=\psi(\U)$. Moreover, let $a\in\F_{q^k}$ be an element with field norm $1$ and consider $\sigma_a=\{(s,as,s^{q^2},\ldots,s^{q^{r-2}}, s^{q^{r-1}}) \mid s\in\F_{q^k}\}$.
We have that 
$$\psi(\sigma_a)=\{(s,as,s,\ldots,s)\mid s\in\F_{q^k}\}=\langle P_a\rangle_{\bF_{q^k}},$$ where $P_a=(1,a,1,\ldots,1)\in \cP'$ has weight 1. Applying the field reduction map one gets that $\cS_1$ and $\cS_2$ are equivalent in the sense of \cref{def:equivalence}.
\end{proof}

\subsection{Comparison with Multilevel Construction}
The multilevel construction also produces constant dimension subspace codes with any given minimum distance value. In this short subsection, we compare such a construction with the one we have proposed.

When $u\leq \frac{rk}{2}$, whenever there exists a maximum scattered linear set, \cref{prop:constr_field_red} provides an intersecting set of maximum size $m_q(k,r,u)=\binom{u}{1}_q=q^{u-1}+q^{u-2}+\dots+q+1$. We now compare this code with one designed via the multilevel construction.
\begin{proposition}
   Let $U\subseteq\bF_q^{rk}$ be a linear subspace of dimension $u\leq \frac{rk}{2}$. Assume that, using the multilevel construction, we construct a constant dimension subspace code $\C$ with minimum distance $2k$ contained in the Schubert variety $\Omega_U$ Then 
   $$|\C|\leq \sum_{1\leq i\leq s}q^{j_i-1},$$ where $s$ is the number of Schubert cells that are selected and for each cell $i$, the quantity $j_i$ denotes the position of the last pivot.
\end{proposition}
\begin{proof}
    First of all, observe that if we want to construct a subspace code with minimum distance $2k$, we can only pick spaces from Schubert cells not sharing a pivot position. The Schubert condition implies the considered cells must have the last pivot in one of the last $u$ coordinates. Hence, we can consider at most $u$ cells. Now, enumerate the last $u$ entries from $1$ to $u$, and assume that we have selected a cell with the last pivot in position $1\leq j\leq u$. This implies that the corresponding Ferrers diagram $\cF$ contains $j-1$ dots in the bottom row, i.e.\ $\nu_{\min}(\cF,k)\leq \nu_{k-1}=j-1$. Hence, by the Singleton-like bound for Ferrers diagram rank-metric codes \eqref{eq:Ferrers_bound}, the number of spaces that we can pick in this cell is at most $q^{j-1}$. 
    Hence, we have that with the multilevel construction we can select $s$ Schubert cells, with $s\leq u$. Moreover, for each cell $1\leq i\leq s$, we have at most $q^{j_i-1}$ spaces, where $j_i\ne j_l$, for every $i\ne l$. We then conclude that the multilevel construction yields a code $\C$ such that 
    $$|\C|\leq \sum_{1\leq i\leq s} q^{j_i-1}.$$
\end{proof}

\begin{corollary}
    Let $u\leq \frac{rk}{2}$. Let $\C$ be a constant subspace code with minimum distance $2k$ in the Schubert variety $\Omega_U$ obtained via the multilevel construction.
    If $r<u$, then $|\C|<\binom{u}{1}_q$.   
\end{corollary}
\begin{proof}
In the multilevel construction, we select different Schubert cells and each cell is uniquely determined by the positions of the pivots in the reduced row echelon form. We choose $k$ pivot positions among $rk$ possibilities. Since all the pivot positions have to be different, we can select at most $\frac{kr}{k}=r$ cells. Hence, the upper bound on the cardinality of $|\C|$ contains at most $r$ terms. Hence, when $r<u$ it gives a quantity strictly smaller than $\binom{u}{1}_q$.
\end{proof}

\begin{remark}
    When $r\geq u$, the multilevel construction might provide a maximum size intersecting set with respect to a $u$-dimensional space. However, when $r<u$, the construction from \cref{prop:constr_field_red} and the one from \cref{cor:construction_norm} for the special case $u=k$ provide a larger intersecting set than the one from the multilevel construction. Hence, we have found cases in which we can design a larger constant dimension subspace code with minimum distance $2k$ in the Schubert variety $\Omega_U$. We illustrate this in the following example. 
\end{remark}

\begin{ex}
Let $q=2, k=u=3, r=2$. 
As we already pointed out, when $r=2$ the constructions from \cref{prop:constr_field_red} and \cref{cor:construction_norm} are equivalent. 
Let $\bF_{2^3}= \F_2(\omega)$, where $\omega^3=\omega+1$. Let $\U=\{(s,s^2)\mid s\in\F_8\}$. Then $\U$ meets each of the spaces $\sigma_a=\{(s,as):s \in \F_8\}$ in a one-dimensional space, where $a$ is an element of $\F_8$ of norm $1$. There are $7$ many spaces $\sigma_a$, which then give rise to a subspace code in the Schubert variety $\Omega_d$, with Schubert condition $d=(3,5,6)$. Instead, the multilevel construction with these parameters only achieves a code size equal to $q^{u-1}+1=5$. 
\end{ex}

\section{A More General Problem}\label{sec:generalization}
In this section, we give a generalization of \cref{def:intersecting_set} and provide some initial bounds on the size of Schubert subspace codes with respect to more general Schubert varieties and that do not necessarily achieve the maximum subspace distance of $2k$.

Let $U$ be a $u$-dimensional subspace of $\bF_q^{rk}$, let $\ell$ be an integer such that $0\leq \ell \leq k$ and define
$$\Omega_{U,\ell} := \{W\in\Gr_q(k,rk) | \dim(W\cap U)\geq \ell\}.$$
For $\ell=1$, we have that $\Omega_{U,1}=\Omega_U$.

In terms of Schubert condition, this corresponds to having $d_\ell=u$. Hence, $\Omega_{U,\ell}$ corresponds to the Schubert variety $\Omega_d\subseteq\Gr_q(k,rk)$ with Schubert condition $d=(u-\ell+1,\dots, u, rk-k+\ell+1,\dots,rk)$.
We are now ready to define a generalization of the concept of intersection set, which will be a Schubert subspace code with respect to this class of Schubert varieties.

\begin{definition} \label{def:lt-intersecting_set}
    Let $k,r,u,\ell,t$ be integers such that $k,r\geq 2$, $u\leq\frac{rk}{2}$, $0\leq \ell\leq k$ and $0\leq t\leq k-1$. Let $U$ be a $u$-dimensional subspace of $\bF_q^{rk}$.  We say that a Schubert subspace code $\mS\subseteq\Omega_{U,\ell}$ is an \textbf{$(\ell,t)$-intersecting set with respect to~$U$} if for every $S_i, S_j\in\mathcal{S}$ s.t.\ $S_i\neq S_j$, it holds that $\dim(S_i\cap S_j)\leq t$.
\end{definition}
We denote by $m_q(k,r,u,\ell,t)$ the largest size of an $(\ell,t)$-intersecting set with respect to $U$.
Notice that a $(1,0)$-intersecting set is an intersecting set according to \cref{def:intersecting_set}.

\begin{remark}
Also in this case it is immediate to see that $m_q(k,r,u,\ell,t)$ does not depend on the choice of $U$, but only on its dimension $u$. If $t>0$ though, this is no longer a partial spread, since different subspaces are allowed to intersect. Similar problems have been considered in \cite{longobardi2022sets}, where the authors studied subspaces that pass through a common space, with only two allowed pairwise intersection dimension values.
\end{remark}

Notice that \cref{def:lt-intersecting_set} allows the subspaces in the code to intersect in dimension up to $t$. Hence, the minimum distance of the subspace code associated with an $(\ell,t)$-intersecting set has to be $\dS(\cS)\geq 2(k-t)$.

We are interested in finding bounds on $m_q(k,r,u,\ell,t)$ and in constructing maximum size $(\ell,t)$-intersecting sets. 
Unfortunately, it is not straightforward to find non-trivial upper bounds on the size $m_q(k,r,u,\ell,t)$. However,  the multilevel construction introduced by Etzion and Silberstein in \cite{etzion2009error} is still a valid construction of $(\ell,t)$-intersecting set that we can use as a lower bound. A large family of subspaces can be found using such construction when the Etzion-Silberstein conjecture turns out to be true. In order to state the bound, let $k,r,u,\ell,t$ be integers with the usual assumptions of this section and consider the Ferrers diagram $\cF$ in \cref{fig:FerrersF}.

\begin{figure}[htp] \begin{tikzpicture}[x=0.6cm,y=0.6cm]

\draw [decorate,decoration={brace,amplitude=5pt,raise=1.5ex}]
  (-0.3,5) -- (6.3,5) node[midway,yshift=1.5em]{$(r-1)k$};
\draw [decorate,decoration={brace,amplitude=5pt,mirror,raise=1.5ex}]
  (-0.3,3) -- (3.5,3) node[midway,xshift=-2em, yshift=-1.5em]{$(r-1)k-(u-\ell)$};
\draw [decorate,decoration={brace,amplitude=5pt,mirror,raise=1.5ex}]
  (3.7,1) -- (6.3,1) node[midway,yshift=-1.5em]{$u-\ell$};
\draw [decorate,decoration={brace,amplitude=5pt,raise=1.5ex}]
  (0,2.7) -- (0,5.3) node[midway,xshift=-2.5em]{$k-\ell$};
\draw [decorate,decoration={brace,amplitude=5pt,raise=1.5ex}]
  (4,0.7) -- (4,2.5) node[midway,xshift=-1.5em,yshift=-0.5em]{$\ell$};
\draw [decorate,decoration={brace,amplitude=5pt,mirror,raise=1.5ex}]
  (6,0.7) -- (6,5.3) node[midway,xshift=1.5em]{$k$};

\draw[fill=black] (0,5) circle (2pt);
\draw[color=black] (0,4.2) node {$\vdots$};
\draw[fill=black] (0,3) circle (2pt);

\draw[color=black] (1,5) node {$\cdots$};
\draw[color=black] (1,3) node {$\cdots$};

\draw[color=black] (3,5) node {$\cdots$};
\draw[color=black] (3,3) node {$\cdots$};

\draw[fill=black] (4,5) circle (2pt);
\draw[color=black] (4,4.2) node {$\vdots$};
\draw[fill=black] (4,3) circle (2pt);
\draw[color=black] (4,2.2) node {$\vdots$};
\draw[fill=black] (4,1) circle (2pt);

\draw[color=black] (5,5) node {$\cdots$};
\draw[color=black] (5,3) node {$\cdots$};
\draw[color=black] (5,1) node {$\cdots$};

\draw[fill=black] (6,5) circle (2pt);
\draw[color=black] (6,4.2) node {$\vdots$};
\draw[fill=black] (6,3) circle (2pt);
\draw[color=black] (6,2.2) node {$\vdots$};
\draw[fill=black] (6,1) circle (2pt);
\end{tikzpicture} 
\caption{Ferrers diagram $\cF$.}
\label{fig:FerrersF}
\end{figure}

\begin{proposition} \label{prop:lower_bound}
Assume the Etzion-Silberstein conjecture holds. Then,
    $$ m_q(k,r,u,\ell,t) \geq q^{\nu_{\min}(\cF,k-t)},$$
    where $\cF$ is the Ferrers diagram in \cref{fig:FerrersF} and $\nu_{\min}(\cF,k-t)=\min\{\nu_0,\nu_{k-t-1}, \nu_{k-\ell}, \nu_{k+\ell-t-u-1}\}$, with 
    \begin{align*}
         \nu_0 &=\begin{cases}
        (rk-k-u+\ell)(k-\ell) + k(t+1-\ell-k+u) & \text{if } t+1\geq \ell+k-u,\\
        (rk-2k+t+1)(k-\ell) & \text{if } t+1< \ell+k-u;
    \end{cases}\\
    \nu_{k-t-1} &=\begin{cases}
        \ell(u-\ell) + (r-1)k(t+1-\ell) & \text{if } t+1\geq \ell,\\
        (t+1)(u-\ell) & \text{if } t+1< \ell;
    \end{cases}\\
      \nu_{k-\ell} &= \begin{cases}
        0 & \text{if } t+1\leq 2\ell-u \text{ and } t+1<\ell, \\
        \ell(u-2\ell+t+1) & \text{if }2\ell-u< t+1 < \ell;
    \end{cases}\\
    \nu_{k+\ell-t-1-u} &= \begin{cases}
        0 & \text{if } t+1\leq 2\ell-u \text{ and } t+1<\ell+k-u , \\
        ((r-1)k-u+\ell)(u+t+1-2\ell) & \text{if } 2\ell-u<t+1 < \ell+k-u.
    \end{cases}
    \end{align*}
\end{proposition}

\begin{proof}
    Without loss of generality, we can do the construction with respect to the standard flag and assume $$U = \rowsp(\boldsymbol{0}_{u\times (rk-u)} \ |\ \mathrm{Id}_u) .$$
    We then note that using the multilevel construction one can construct a constant dimension subspace code $\cS$ with minimum distance value $2(k-t)$, contained in the Schubert variety $\Omega_{U,\ell}$.  In our language, this corresponds to an $(\ell,t)$-intersecting set with respect to $U$.     
    With respect to the standard flag, we can write the Schubert variety $\Omega_{U,\ell}$ as the one generated by the following cell, written in echelon-Ferrers form:
\[
M =
\left( \begin{array}{@{}c|c@{}}
	\begin{array}{@{}c|c@{}}
		\begin{array}{c|c@{}}
			\mathrm{Id}_{k-\ell} & \bullet_{(k-\ell)\times((r-1)k-u+\ell)} 
		\end{array} 
	& \boldsymbol{0}_{(k-\ell)\times \ell} \\
	\cmidrule[0.2pt]{1-2}
	\boldsymbol{0}_{\ell\times(rk-u)} & \mathrm{Id}_\ell \\
	\end{array}
& \bullet_{k\times(u-\ell)}
\end{array} \right).
\]
Since all other cells of the Schubert variety correspond to Ferrers diagrams with fewer dots, the cardinality of a Ferrers diagram rank-metric code constructed on the largest cell is one of the largest. Hence, in the following, we will only focus on subspaces in such largest cell. This bounds from below the size of a code constructed over all cells of the Schubert variety.

Among the subspaces with echelon-Ferrers form $M$, we now want to select a large number of subspaces that are $(\ell,t)$-intersecting sets according to \cref{def:lt-intersecting_set}, i.e.\ which pairwise intersect in a space of dimension at most $t$. This condition means that the subspace distance between any two subspaces $S_i,S_j$ is 
$$\dd_S(S_i,S_j) \geq 2(k-t).$$
Since we are restricting ourselves to one Schubert cell, we can study rank-metric codes supported on the corresponding Ferrers diagram. It follows that $m_q(k,r,u,\ell,t)\geq q^{\dim(\cF,\delta)}$, where $\dim(\cF,\delta)$ is the maximum dimension of a code with minimum rank distance
$$ \delta = k-t $$
and supported on the Ferrers diagram $\cF$ coming from the echelon-Ferrers form $M$, which is a $k\times(r-1)k$ Ferrers diagram such that the top $k-\ell$ rows have $(r-1)k$ dots, and the remaining $\ell$ rows have $k-\ell$ dots. See \cref{fig:FerrersF}.

We now compute $\nu_{\min}(\cF,k-t)=\min_i\{\nu_i\}$, where for $0\leq i\leq k-t-1$, $\nu_i$ is the number of dots that are not contained in the first $i$ rows and are not contained in the rightmost $k-t-1-i$ columns. Analogously, $\nu_i$ is the number of dots contained both in the bottom $k-i$ rows and in the leftmost $(r-2)k+t+1+i$ columns. In particular, $\nu_0$ is the number of dots in the leftmost $(r-2)k+t+1$ columns, while $\nu_{k-t-1}$ is the number of dots in the bottom $t+1$ rows.

Hence:
\begin{align*}
    \nu_0&=\begin{cases}
        (rk-k-u+\ell)(k-\ell) + k(t+1-\ell-k+u) & \text{if } t+1\geq \ell+k-u,\\
        (rk-2k+t+1)(k-\ell) & \text{if } t+1< \ell+k-u;
    \end{cases}\\
    \nu_{k-t-1} &=\begin{cases}
        \ell(u-\ell) + (r-1)k(t+1-\ell) & \text{if } t+1\geq \ell,\\
        (t+1)(u-\ell) & \text{if } t+1< \ell.
    \end{cases}
\end{align*}
If $t+1\geq\ell$, computations show that $\nu_0\geq\nu_{k-t-1}$, hence $\nu_0$ can not be the minimum $\nu_i$.

Due to the shape of the Ferrers diagram, if $t+1$ is large enough the minimum $\nu_i$ is obtained for either the first or last index.

If $t+1<\ell$ though, decreasing the index from $i=k-t-1$ does not necessarily imply that $\nu_i$ also counts dots from the top $k-\ell$ longer rows. The last index for which this does not happen is $i=k-\ell$. In this case, we get:
\begin{equation*}
    \nu_{k-\ell} = \begin{cases}
        0 & \text{if } t+1\leq 2\ell-u \text{ and } t+1<\ell, \\
        \ell(u-2\ell+t+1) & \text{if }2\ell-u< t+1 < \ell. 
    \end{cases}
\end{equation*}
Comparing this value with the one of $\nu_{k-t-1}$ yields
$$ \min\{\nu_{k-\ell},\nu_{k-t-1}\} = \begin{cases}
    0 & \text{if }\ell\leq u< 2\ell \text{ and }t+1\leq 2\ell-u, \\
    \ell(u-2\ell+t+1) & \text{if }\ell\leq u< 2\ell \text{ and }2\ell-u\leq t+1<\ell, \\
    (t+1)(u-\ell) & \text{if }2\ell\leq u \text{ and }t+1<\ell, \\  
    \ell(u-\ell) + (r-1)k(t+1-\ell) & \text{if }t+1\geq\ell.
\end{cases}$$

Analogously, if $t+1<\ell+k-u$, increasing the index from $i=0$ does not necessarily imply that $\nu_i$ also counts dots from the right $u-\ell$ longer columns. The last index for which this does not happen is $i=k+\ell-t-1-u$. In this case, we get:
\begin{equation*}
    \nu_{k+\ell-t-1-u} = \begin{cases}
        0 & \text{if } t+1\leq 2\ell-u \text{ and } t+1<\ell+k-u , \\
        ((r-1)k-u+\ell)(u+t+1-2\ell) & \text{if } 2\ell-u<t+1 < \ell+k-u; 
    \end{cases}
\end{equation*}
Combining all the above arguments, we get the statement.
\end{proof}

In order to provide new constructions of $(\ell,t)$-intersecting sets which are comparable with the multilevel one, we need to build codes whose size is at least $q^{\dim(\mathcal{F},\delta)}$.
Clearly, in case of a $(1,0)$-intersecting set with respect to a $u$-dimensional space $U$ with $u\leq \frac{kr}{2}$, we have an order of magnitude equal to $q^{u-1}$. In this case, we can be more precise and compare the multilevel construction with the one we proposed in \cref{sec:prob_formulation}.

\subsection{An upper bound on \texorpdfstring{$m_q(k,r,u,\ell,t)$}{m(k,r,u,l,t)}}
As we already pointed out above, it is not easy in general to provide a non-trivial upper bound for $m_q(k,r,u,\ell,t)$. However, in the case when $t+1\leq\ell$ we can give the following easy concrete upper bound as a generalization of \cref{thm:upp_bound}.

\begin{proposition} \label{prop:upper_bound}
    If $t+1\leq\ell$, then $$m_q(k,r,u,\ell,t) \leq \binom{u}{\ell}_q.$$
\end{proposition}
\begin{proof}
Let $t+1\leq\ell$ and let $\cS$ be an $(\ell,t)$-intersecting set with respect to a $u$-dimensional subspace $U$. Then every $S\in\cS$ intersects $U$ in an $\ell$-dimensional space and no two $S_i,S_j\in\cS$ can intersect $U$ in the same $\ell$-dimensional space since they intersect each other in a space of dimension at most $t$ and $t<\ell$. Since $U$ contains exactly $\binom{u}{\ell}_q$ many $\ell$-dimensional subspaces, we conclude.
\end{proof}

\begin{remark}
    Let $\ell=1, t=0$. Then,  the upper bound from \cref{prop:upper_bound} is tight. Indeed, the constructions given in \cref{sec:prob_formulation} provide maximum intersecting sets of size $\binom{u}{1}_q$ with respect to a $u$-dimensional subspace of $\F_q^{rk}$.
\end{remark}
\begin{remark}
    Let $\ell=t+1$. \cref{prop:lower_bound} yields the lower bound $m_q(k,r,u,\ell,\ell-1)\geq q^{\ell(u-\ell)}$. In this case, the order of magnitude of this lower bound is the same as the one of the upper bound of \cref{prop:upper_bound}.
\end{remark}

\section{Conclusion and Open Questions}\label{sec:conclusion}
In this paper we have introduced for the first time, to the best of our knowledge, the problem of constructing constant dimension subspace codes restricted to Schubert varieties. Moreover, we have introduced the notion of Schubert subspace codes with respect to a Schubert variety $\Omega_d$ and of $(\ell,t)$-intersecting sets in $\Gr_q(k,rk)$ with respect to a given subspace $U$ of the vector space $\F_q^{rk}$. This is essentially a geometrical translation of the same problem, which, however, looks more practical to use. In general, we are interested in large intersecting sets. When the parameter $t$ is equal to $0$, i.e.\ when the considered spaces are all disjoint, we provide an upper bound for the size of $(1,0)$-intersecting sets with respect to a $u$-dimensional space $U$, with $u\leq \frac{rk}{2}$. We moreover give a construction of sets achieving the upper bound. The general case is essentially open, but we give the order of magnitude that the intersecting set should have in order to be comparable with one obtained via the known multilevel construction. 
Our study gives rise to plenty of further research questions. We list a few of them.

\begin{enumerate}
    \item We considered so far subspace codes contained in $\Gr_q(k,rk)$. This choice is due to the technicality of the construction in \cref{prop:constr_field_red}, which requires applying the field reduction map. However, it is possible to consider more generally $\Gr_q(k,n)$, where $n$ is not necessarily a multiple of $k$. When $\ell=1$ and $t=0$, while the upper bound from \cref{thm:upp_bound} stays the same, the construction from \cref{prop:constr_field_red} can no longer be applied. Can we find an alternative construction of maximum size intersecting sets with respect to a $u$-dimensional subspace of $\F_q^n$? 
    \item The construction given in \cref{prop:constr_field_red} makes use of the concept of scattered linear sets. A necessary condition for the existence of these objects is that they arise from $u$-dimensional $q$-systems, with $u\leq\frac{rk}{2}$. This limits our study only to some Schubert varieties. Is it also possible to provide a construction that does not use linear sets and also covers the case $u>\frac{rk}{2}$? 
    \item In \cref{prop:transformation}, we show that it is possible to obtain the construction given in \cref{cor:construction_norm} as field reduction of points of weight one of a linear set. It remains an open question to show whether the construction from \cref{prop:constr_field_red} and the one in \cref{cor:construction_norm} are or are not equivalent according to \cref{def:equivalence}.
    \item Is it possible to give an upper bound for $m_q(k, r, u, \ell, t)$ which is more general than the one given in \cref{prop:upper_bound}? 
    \item Is it possible to generalize the construction from \cref{prop:constr_field_red} in order to construct $(\ell,t)$-intersecting sets for general parameters?  
\end{enumerate}

\bigskip
\bigskip
\section*{Acknowledgements}
The authors are thankful to Alessandro Neri and John Sheekey for fruitful discussions.\\
Gianira~N. Alfarano is supported by the Swiss National Foundation through grant no. 210966. Joachim Rosenthal and Beatrice Toesca are supported by the Swiss National Foundation through grant no. 212865.

\bibliographystyle{abbrv}
\bibliography{bibliography.bib}

\end{document}